\documentclass[11pt, onecolumn]{IEEEtran}

\usepackage{amsthm}
\usepackage{amsmath}
\usepackage{amssymb}
\usepackage{amsfonts}
\usepackage{epsfig}
\usepackage{subfigure}
\usepackage{calc}
\usepackage{color}
\usepackage[all]{xy}
\usepackage{bm}
\usepackage{algorithmicx}
\usepackage{algpseudocode}
\usepackage{algorithm}
\usepackage{enumerate}
\usepackage{hyperref}
\usepackage{float}

\newtheorem{defn}{Definition}
\newtheorem{thm}{Theorem}[section]
\newtheorem{cor}[thm]{Corollary}
\newtheorem{prop}{Proposition}

\newtheorem{lem}[thm]{Lemma}
\newtheorem{conj}[thm]{Conjecture}
\newtheorem{constr}[thm]{Construction}
\newtheorem{note}{Remark}

\newcommand{\bit}{\begin{itemize}}
\newcommand{\eit}{\end{itemize}}
\newcommand{\bcor}{\begin{cor}}
\newcommand{\ecor}{\end{cor}}
\newcommand{\beq}{\begin{equation}}
\newcommand{\eeq}{\end{equation}}
\newcommand{\beqn}{\begin{equation*}}
\newcommand{\eeqn}{\end{equation*}}
\newcommand{\bea}{\begin{eqnarray}}
\newcommand{\eea}{\end{eqnarray}}
\newcommand{\bean}{\begin{eqnarray*}}
\newcommand{\eean}{\end{eqnarray*}}
\newcommand{\ben}{\begin{enumerate}}
\newcommand{\een}{\end{enumerate}}
\newcommand{\bdefn}{\begin{defn}}
\newcommand{\edefn}{\end{defn}}
\newcommand{\bnote}{\begin{note}}
\newcommand{\enote}{\end{note}}
\newcommand{\bprop}{\begin{prop}}
\newcommand{\eprop}{\end{prop}}
\newcommand{\blem}{\begin{lem}}
\newcommand{\elem}{\end{lem}}
\newcommand{\bthm}{\begin{thm}}
\newcommand{\ethm}{\end{thm}}
\newcommand{\bconj}{\begin{conj}}
\newcommand{\econj}{\end{conj}}
\newcommand{\bconstr}{\begin{constr}}
\newcommand{\econstr}{\end{constr}}
\newcommand{\bpf}{\begin{proof}}
\newcommand{\epf}{\end{proof}}

\newcommand{\uv}{{\bf u}}

\newcommand{\vv}{{\bf v}}

\newcommand{\Um}{{\bf U}}

\newcommand{\Vm}{{\bf V}}

\newcommand{\Bc}{{\cal B}}
\newcommand{\Cc}{{\cal C}}
\newcommand{\Dc}{{\cal D}}

\newcommand{\Kc}{{\cal K}}

\newcommand{\Nc}{{\cal N}}

\newcommand{\Xc}{{\cal X}}

\newcommand{\field}[1]{\mathbb{#1}}
\newcommand{\F}{\field{F}}

\newcommand{\cC}{{\cal C}}

\begin{document}
\title{Explicit MBR All-Symbol Locality Codes}


\author{Govinda M. Kamath, Natalia Silberstein, N. Prakash, Ankit S. Rawat, \\ V. Lalitha,  O. Ozan Koyluoglu, P. Vijay Kumar, and Sriram Vishwanath
\thanks{Govinda M. Kamath,  N. Prakash, V. Lalitha and P. Vijay Kumar are with the Department of ECE, Indian Institute of Science, Bangalore, India (email: \{govinda,  prakashn, lalitha, vijay\}@ece.iisc.ernet.in).
Natalia Silberstein, Ankit S. Rawat, O. Ozan Koyluoglu, and Sriram Vishwanath are with Department of ECE, University of Texas at Austin, Austin, USA (email: \{natalys, ankitsr, ozan, sriram\}@austin.utexas.edu). This research is supported in part by the National Science Foundation under Grant 0964507 and in part by the NetApp Faculty Fellowship program. The work of V. Lalitha is supported by a TCS Research Scholarship.
} }

\maketitle

\begin{abstract}
Node failures are inevitable in distributed storage systems (DSS). To
enable efficient repair when faced with such failures, two main techniques are known: Regenerating codes, i.e., codes that
minimize the total repair bandwidth; and codes with locality, which minimize the number of nodes participating in the repair process.
This paper focuses on regenerating codes with locality, using pre-coding based on Gabidulin codes, and presents constructions
that utilize minimum bandwidth regenerating (MBR) local codes. The constructions
achieve maximum resilience (i.e., optimal minimum distance) and have maximum capacity (i.e., maximum rate). Finally, the same
pre-coding mechanism can be combined with a subclass of fractional-repetition codes to enable maximum resilience and
repair-by-transfer simultaneously.
\end{abstract}

\section{Background}
\label{sec:background}

\subsection{Vector Codes}

An $[n, K, d_{\text{min}}, \alpha]$ \textit{vector code} over a field $\mathbb{F}_q$ is a code ${\mathcal{C}}$ of block length $n$, having a symbol alphabet $\mathbb{F}_q^{\alpha}$ for some $\alpha>1$, satisfying the additional property that given $\mathbf{c}, \mathbf{c}' \in \mathcal{C}$ and $a,b \in \mathbb{F}_q$, $a\mathbf{c} + b\mathbf{c}'$ also belongs to $\mathcal{C}$.  As a vector space over $\mathbb{F}_q$, $\mathcal{C}$ has dimension $K$, termed the scalar dimension (equivalently, the file size) of the code and as a code over the alphabet $\mathbb{F}_q^{\alpha}$, the code has minimum distance $d_{\min}$.

Associated with the vector code $\mathcal{C}$ is an $\mathbb{F}_q$-linear scalar code $\mathcal{C}^{(s)}$ of length $N= n\alpha$, where $\mathcal{C}^{(s)}$ is obtained by expanding each vector symbol within a codeword into $\alpha$ scalar symbols (in some prescribed order).   Given  a generator matrix $G$ for the scalar code ${\cal C}^{(s)}$, the first code symbol in the vector code is naturally associated with the first $\alpha$ columns of $G$ etc.  We will refer to the collection of $\alpha$ columns of $G$ associated with the $i^{\text{th}}$ code symbol ${\bf c}_i$ as the $i^{\text{th}}$ thick column and to avoid confusion, the columns of $G$ themselves as thin columns.

\subsection{Locality in Vector Codes} \label{sec:locality_basics}

Let $\mathcal{C}$ be an $[n, K, d_{\text{min}}, \alpha]$ vector code over a field $\mathbb{F}_q$, possessing a $(K \times n\alpha)$ generator matrix $G$.  The $i^{\text{th}}$ code symbol, ${\bf c}_i$, is said to have  $(r, \delta)$ locality, $\delta \geq 2$, if there exists a punctured code $\Cc_i:=\Cc|_{S_i}$ of
${\mathcal C}$ (called a \emph{local code}) with support $S_i\subseteq \{1, 2, \cdots, n\}$ such that
\begin{itemize}
\item $i \in S_i$,
\item $|S_i| \leq n_L := r + \delta - 1$, and
\item $d_{\text{min}}\left(\mathcal{C}|_{S_i}\right) \geq \delta$.
\end{itemize}

The code $\mathcal{C}$ is said to have $(r,\delta)$ \emph{information locality} if there exists $l$ code symbols with $(r, \delta)$ locality and respective support sets $\{S_i\}_{i=1}^{l}$ satisfying
\begin{itemize}
\item  $\displaystyle\text{Rank}(G|_{\cup_{i=1}^{l}S_i})=K$.
\end{itemize}
The code $\mathcal{C}$ is said to have $(r,\delta)$ \emph{all-symbol locality} if all code symbols have $(r,\delta)$ locality. A code with $(r,\delta)$ information (respectively, all-symbol) locality is said to have \emph{full} $(r,\delta)$ information (respectively, all-symbol) locality, if all local codes have parameters given by $|S_i| = r + \delta - 1 $ and $d_{\text{min}}\left(\mathcal{C}_i\right) = \delta$, for $i=1, \cdots, l$.


The concept of locality for scalar codes, with $\delta=2$, was introduced in \cite{GopHuaSimYek} and extended  in \cite{PraKamLalKum} and \cite{PapDim} to scalar codes with arbitrary $\delta$, and vector codes with $\delta=2$, respectively.  This was further extended to vector codes with arbitrary $\delta$ in \cite{KamPraLalKum} and \cite{RawKoySilVis}, where, in addition to constructions of vector codes with locality, authors derive minimum distance upper bounds and also consider settings in which the local codes have regeneration properties.

Consider now a vector code $\Cc$ with full $(r,\delta)$
locality whose associated local codes $\Cc_i$ have parameters $[n_L,K_L,\delta]$. In this paper, we are interested in local codes that have the uniform rank accumulation property, in particular, local MBR codes and local fractional-repetition codes.
\begin{defn}[Uniform rank accumulation (URA) codes]
Let $G$ be a generator matrix for a code $\mathcal{C}$, and $S_i$ be an arbitrary subset of $i$ thick columns of $G$, for some $i=1,\cdots,n$. Then, $\Cc$ is an URA code, if the restriction $G |_{S_i}$ of $G$ to $S_i$, has rank $\rho_i$ that is independent of the specific subset $S_i$ of $i$ indices chosen and given by
$ \rho_i  =  \sum_{j=1}^{i} a_j$
for some set of non-negative integers $\{a_j\}$.
\end{defn}
We will refer to the sequence $\{a_i, 1 \leq i \leq n\}$ as the \emph{rank accumulation profile} of the code $\mathcal{C}$.

We now present the minimum distance upper bound given in~\cite{KamPraLalKum} for the case when local codes are URA codes.
Consider the finite length vector $(a_1,a_2,\cdots,a_{n_L})$, and its extension to a periodic semi-infinite sequence $\{a_i\}_{i=1}^{\infty}$ of period $n_L$ by defining
$a_{i+j n_L}  =  a_{i}, \ 1 \leq i \leq n_L, \  j \geq 1.$
Let $P(\cdot)$ denote the sequence of partial sums,
\bea
P(s) & = & \sum_{i=1}^s a_i \label{eq:P_fun}, \ \ s\geq 1.
\eea
Then, given integers $ u_1 \geq 0, \ \ 1 \leq u_0 \leq n_L $,
$P(u_1n_{L}+u_0)  =  u_1K_L+P(u_0).$
Next, let us  define the function $P^{(\text{inv})}$ by setting $P^{(\text{inv})}(\nu)$, for $\nu \geq 1 $, to be the smallest integer $s$ such
that $P(s) \geq \nu$.
It can be verified that for $ v_1\geq 0$ and $1 \leq v_0 \leq K_L$,
\bean
P^{(\text{inv})}(v_1K_L+v_0)  & = & v_1n_L+P^{(\text{inv})}(v_0),
\eean
where $P^{(\text{inv})}(v_0) \leq r$ as $1 \leq v_0 \leq K_L$.

The minimum distance of a code $\Cc$ whose local codes $\Cc_i$ are URA codes can be bounded as follows.
\begin{thm}[Theorem 5.1 of \cite{KamPraLalKum}] \label{thm:URA_bound}
The minimum distance of ${\cal C}$ is upper bounded by
\bea \label{eq:dmin_Pinv_bound}
d_{\min} & \leq & n-P^{(\text{inv})}(K)+1.
\eea
\end{thm}

The codes achieving the bound in \eqref{eq:dmin_Pinv_bound} are referred to
as codes having optimal locality. For such locality optimal codes, one can
then analyze whether the code allows for efficient data storage in DSS.
Towards this end, file size bound for codes with locality are given in~\cite{RawKoySilVis}
using the min-cut techniques similar to that of~\cite{PapDim}. As noted in~\cite{KamPraLalKum},
when URA codes are used as local codes, the file size bound for $d_{\min}$-optimal codes
can be represented in the form
\begin{eqnarray}
K &\leq&  P(n-d_{\min}+1) \nonumber \\
&=& \left(\left\lceil \frac{n-d_{\min}+1}{n_L} \right\rceil - 1 \right)K_L+ P(l_0) \label{eq:FileSizeBound-Local},
\end{eqnarray}
where $l_0\in\{1,\cdots, n_L\}$ is such that
$$n-d_{\min}+ 1 = \left(\left\lceil \frac{n-d_{\min}+1}{n_L} \right\rceil - 1 \right)n_L + l_0.$$
We note that $P(l_0)=P(r)$, for $r\leq l_0$.

\subsection{MBR Codes}

An $((n,k,d), (\alpha,\beta), K)$ \emph{minimum-bandwidth regenerating} (MBR) code is an $[n,K,d_{\min}=n-k+1,\alpha]$ vector code satisfying additional constraints described below. The code is intended to be used in a distributed storage network in which each code symbol is stored within a distinct node.  The code is structured in such a way that the entire file can be recovered by processing the contents of any $k$, $1 \leq k \leq n$ nodes. Further, in case of a single node failure, the replacement node can reconstruct the data stored in the failed node by connecting to any $d$, $k \leq d \leq n-1$,  nodes and downloading $\beta=\frac{\alpha}{d}$ symbols from each node.  The scalar dimension (or file size) parameter $K$ can be expressed in terms of the other parameters as:
\bean
K & = & \left( dk-{k \choose 2} \right) \beta,
\eean
as proved in~\cite{DimGodWuWaiRam}.
A cut-set bound derived from network coding shows us that the file size cannot be any larger, and thus, MBR codes are example of regenerating codes that are optimal with respect to file size.  A regenerating code is said to be exact if the replacement of a failed node stores the same data as did the failed node, and functional otherwise. We are concerned here only with exact-repair codes. Constructions of MBR codes for all $k \leq  d =\alpha  < n$ and $\beta=1$ are presented in \cite{RasShaKum_pm}. MBR codes with repair by transfer and $d=n-1$ are presented in \cite{ShaRasKumRam_rbt}.

It can be inferred from the results in~\cite{ShaRasKumRam_rbt} that MBR codes are URA codes. In particular, for an $((n,k,d), (\alpha,\beta),K)$ MBR code, the rank accumulation profile is given by
\bea
a_j =  \left\{ \begin{array}{c c}
              \alpha - (j-1)\beta, & 1 \leq j \leq k \\
              0, & k+1 \leq j \leq n. \label{eq:MBR_rankacc}
	     \end{array} \right.
\eea

\subsection{MBR-Local Codes}

Let ${\cal C}$ be an $[n,K,d_{\min},\alpha]$ vector code with
\bit
\item full $(r, \delta)$-information locality with $\delta \geq 2$, and
\item all of whose associated local codes $\mathcal{C}_i, i \in \mathcal{L}$ are MBR codes
with identical parameters $((n_L=r+\delta-1,r,d), (\alpha,\beta),K_L)$.
\eit
Then, the dimension of each local code is given by
\bea
K_L & = & \sum_{i = 1}^{n_L}a_i \ = \ \alpha r - {r \choose 2} \beta,
\eea
where $\{a_i, 1 \leq i \leq n_L\}$ is the rank accumulation profile of the MBR code $\mathcal{C}$.

\subsubsection{Minimum distance bound for MBR-Local Codes}

As MBR codes are URA codes, from Theorem~\ref{thm:URA_bound}, we have
\begin{equation}
d_{\min} \leq n-P^{(\text{inv})}(K)+1, \label{eq:dmin-MBR-Local}
\end{equation}
where, for MBR codes we have
\begin{equation}
P^{(\text{inv})} (v_1K_L + v_0) =  v_1n_L + \nu
\end{equation}
for some $v_1 \geq 0$, $1 \leq v_0 \leq K_L$, and
$\nu$ is uniquely determined from
$ \alpha (\nu-1)  - {\nu-1 \choose 2}\beta < v_0 \leq \alpha \nu  - {\nu \choose 2}\beta$.

\subsubsection{File size bound for MBR-Local Codes}

From \eqref{eq:FileSizeBound-Local}, the file size bound for an optimal locality code with MBR local codes is given by
\begin{eqnarray}
K \leq \left(\left\lceil \frac{n-d_{\min}+1}{n_L} \right\rceil - 1 \right)K_L+ \alpha \mu - {\mu \choose 2} \beta \label{eq:FileSizeBound-MBR},
\end{eqnarray}
where $\mu=\min\{l_0,r\}$ with $l_0$ as defined in Subsection~\ref{sec:locality_basics}. Note that \eqref{eq:FileSizeBound-MBR} follows from the rank accumulation profile of MBR codes, i.e., from \eqref{eq:MBR_rankacc}.

\subsection{Linearized Polynomials}

 A polynomial $f(x)$ over the field $\mathbb{F}_{q^m}$, is said to be {\em linearized} of $q$-degree $t$, if
 \begin{equation}
  f(x)= \displaystyle\sum_{i=0}^{t} u_i x^{q^i} \ , \ u_i \in \mathbb{F}_{q^m}, \ u_t \neq 0.
 \end{equation}

A linearized polynomial $f(x)$ over $\mathbb{F}_{q^m}$ satisfies the following property~\cite{MacSlo}:
\begin{eqnarray}
 f(\lambda_1 \theta_1+  \lambda_2 \theta_2) &=& \lambda_1f(\theta_1) + \lambda_2f(\theta_2) \nonumber \\
  & & \forall \ \theta_1, \theta_2 \in \mathbb{F}_{q^m}, \  \lambda_1 ,\lambda_2 \in \mathbb{F}_q . \label{eq:linearity}
\end{eqnarray}

A linearized polynomial $f(x)$ over $\mathbb{F}_{q^m}$ of $q$-degree $t$, $m > t$, is uniquely determined from its evaluation at a set of ${(t+1)}$ points $g_1, \cdots, g_{t+1} \in \mathbb{F}_{q^m}$, that  are linearly independent over $\F_q$.

\subsection{Gabidulin Maximum Rank Distance Codes}

Now, we present a construction of maximum rank distance codes, provided by Gabidulin in~\cite{Gab}. This codes can be viewed as a rank-metric analog of Reed-Solomon codes.


The \emph{rank} of a vector $\vv\in\F_{q^m}^{\Nc}$, denoted by $\rm{rank}(\vv)$ is defined as the rank of the $m\times \Nc$ matrix $\Vm$ over $\F_q$, obtained by expansion of every entry of $\vv$ to a column vector in $\F_q^m$, based on the isomorphism between $\F_{q^m}$ and $\F_{q}^m$. Similarly, for two vectors $\vv,\uv \in \F_{q^m}^{\Nc}$, the \emph{rank distance} is defined by $d_R(\vv,\uv)=\rm{rank}(\Vm - \Um)$.

An $[\Nc,\Kc,\Dc]_{q^m}$ \textmd{rank-metric code} $\cC\subseteq\F_{q^m}^{\Nc}$ is a linear block code over $\F_{q^m}$ of length $\Nc$, dimension $\Kc$  and minimum rank distance $\Dc$. A rank-metric code that attains the Singleton bound $\Dc\leq \Nc-\Kc+1$ in rank-metric is called a \emph{maximum rank distance} (MRD) code.
For $m\geq \Nc$, a construction of MRD codes, called Gabidulin codes is given as follows~\cite{Gab}.

A codeword in an $[\Nc,\Kc,\Dc=\Nc-\Kc+1]_{q^m}$ Gabidulin code
$\cC^{\rm{Gab}}$, $m\geq \Nc$, is defined as 
\begin{equation}
\mathbf{c} = (f(\theta_1),f(\theta_2),\ldots, f(\theta_{\Nc})) \ \in \ \F_{q^m}^{\Nc},
\end{equation} 
where $f(x)$ is a
linearized polynomial over $\F_{q^m}$ of $q$-degree at most $\Kc-1$ with the coefficients given by the information message, and where the $\theta_1,\ldots, \theta_{\Nc}\in \F_{q^m}$ are linearly independent over $\F_q$~\cite{Gab}.

\section{Construction of Codes with MBR Locality} \label{sec:MBR_locality_code_constn}

In this section, we will present two constructions of codes with local regeneration.  In both cases, the local codes are MBR codes with identical parameters and both codes are optimal, i.e., they achieve the upper bound of Theorem~\ref{thm:URA_bound} on minimum distance.  The first construction is an all-symbol locality construction, while the second has information locality.

The constructions presented in this paper, adopt the linearized polynomial approach made use of in \cite{OggDat,SilRawVis,RawKoySilVis}. In particular, similar to the constructions proposed in \cite{SilRawVis,RawKoySilVis}, the constructions of this paper have a two-step encoding process with the first step utilizing Gabidulin codes, which in turn, are based on linearized polynomials. The first code construction given below also proves the tightness of the bound on minimum distance of  codes with URA derived in \cite{KamPraLalKum} (Theorem 5.1) for the case when $K_L\nmid K$, where $K_L$ is the scalar dimension of the local MBR code.

Consider a code ${\cal C}_{\text{\tiny BASIC}}$ that is simply the concatenation of $t$ local MBR codes having identical parameters $((n_L,k,d), (\alpha,\beta),K_L)$.   Thus a typical codeword ${\bf c} \in {\cal C}_{\text{\tiny BASIC}}$ looks like
\bean
{\bf c} & = & \left( \begin{array}{cccc} {\bf c}_1^{\rm mbr} &  {\bf c}_2^{\rm mbr} & \cdots &  {\bf c}_t^{\rm mbr} \end{array} \right ),
\eean
where each vector ${\bf c}_i^{\rm mbr}$ is a codeword belonging to the MBR code.
The generator matrix $G_{\text{\tiny BASIC}}$ of the code will clearly have a block-diagonal structure. It is straightforward to show that the smallest number $\rho$, such that any $\rho$ thick columns of $G_{\text{\tiny BASIC}}$ have rank $\geq K$ is given by $P^{(\text{inv})}(K)$, for any $1 \leq K \leq t K_L$.

\vspace*{0.05in}

\begin{constr}\label{constr:all_symbol_MBR}  We will describe the construction by showing how encoding of a message vector takes place. The encoding is illustrated in Fig.~\ref{fig:mbr_allsymbol}. Given the message vector ${\bf u} \in \mathbb{F}_{q^m}^K$, we first encode ${\bf u}$ to a $tK_L$ long Gabidulin codeword using $tK_L$ linearly independent points (over $\mathbb{F}_{q}$) $\{\theta_{1}, \theta_{2},\ldots, \theta_{tK_L}\} \subset \mathbb{F}_{q^m}$, i.e., by applying an $[tK_L,K, tK_L-K+1]_{q^m}$ Gabidulin code, assuming $m \geq tK_L$. We then partition $tK_L$ symbols of the Gabidulin codeword, $(f(\theta_1), f(\theta_2),\ldots, f(\theta_{tK_L}))$, into $t$ disjoint sets of $K_L$ symbols each. Each of these sets is then fed in as a message vector to a bank of $t$ identical MBR encoders whose outputs constitute $((n_L,r,d),(\alpha,\beta),K_L)$ MBR codes.  If $\left\{ {\bf c}_i^{\rm mbr} \mid i=1,2,\ldots,t\right\}$ is the resulting set of $t$ codewords, these codewords are then concatenated to obtain the desired codeword ${\bf c}$. The code $\mathcal{C}$ thus constructed has:
\begin{itemize}
 \item length $n=tn_L$
 \item $t$ local $((n_L,r,d),(\alpha,\beta), K_L)$  MBR  codes with disjoint supports
 \item full $(r,\delta)$ all-symbol locality where $\delta$ is defined from $n_L=r+\delta-1$.
\end{itemize}
\end{constr}

\begin{figure*}[ht]
\begin{center}
\includegraphics[height=2in]{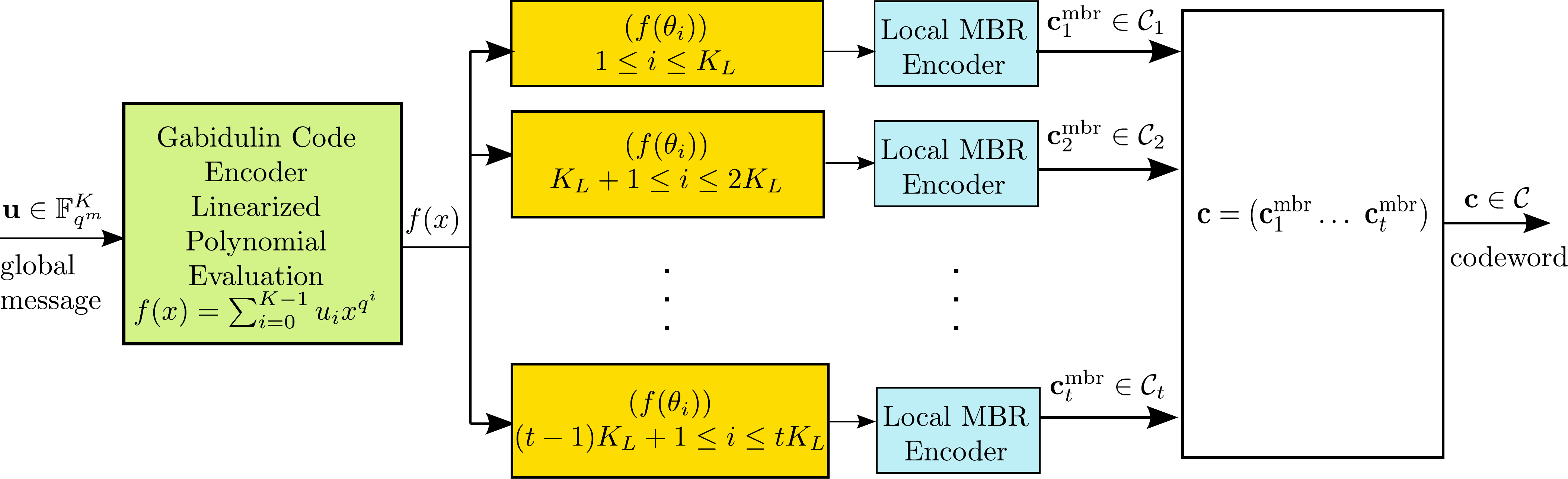}
\caption{Illustrating the two-step construction of the all-symbol MBR-local code.}
\label{fig:mbr_allsymbol}
\end{center}
\end{figure*}

\begin{thm}\label{thm:all_symbol_MBR}

 Given any set of parameters $n,r,\delta,K$, such that $n=tn_L$ and $K \leq tK_L$, the construction \ref{constr:all_symbol_MBR}, yields an optimal MBR-local code with full $(r,\delta)$ all-symbol locality whose minimum distance is given by
\bean
d_{\min} =  n \ - \ P^{\text{(inv)}} (K) +1.
\eean
\end{thm}

\vspace*{0.05in}

We first present a useful lemma on codes that are obtained by concatenating a Gabidulin code (over $\mathbb{F}_{q^m}$) with a vector code (over $\mathbb{F}_{q}$).

\begin{lem} \label{lem:useful}
Let $G$ be the generator matrix of an $[n, J, d_{\min},\alpha]$ vector code over the field $\mathbb{F}_q$. Let $\tilde{J}$ be an integer such that $\tilde{J} \leq J$. Let $\rho$ be the smallest integer, such that the submatrix of $G$ obtained by selecting any $\rho$ thick columns of $G$ results in a matrix of rank $\geq \tilde{J}$.  Let
\bean
f(x) = \sum_{i=0}^{\tilde{J}-1}u_ix^{q^i}, u_i \in \mathbb{F}_{q^m}, \ \ m>J ,
\eean
be a linearized polynomial of $q$-degree at most $\tilde{J}-1$ over the extension field $\mathbb{F}_{q^m}$, for $m \geq J$.
Let $\{\theta_i\}_{i=1}^J$ be any collection of $J$ elements of $\mathbb{F}_{q^m}$ that
are linearly independent over $\F_q$.  The mapping
\bean
(u_0,u_1,\cdots, u_{\tilde{J}-1}) & \rightarrow & (f(\theta_1),f(\theta_2), \cdots , f(\theta_J))G
\eean
defines a linear code ${\cal C}$ over $\mathbb{F}_{q^m}$ having message vector $(u_0, u_1,\cdots, u_{\tilde{J}-1})$.  Then ${\cal C}$ has minimum distance $D_{\min}$ given by
\bean
D_{\min} & = & n-\rho+1,
\eean
i.e., $\mathcal{C}$ is an $[n,\tilde{J},D_{\min}]$ code over $\mathbb{F}_{q^m}$.
\end{lem}

\vspace*{0.05in}

\begin{proof}  Since $f(\cdot)$ is linearized, we can interchange linear operations with the operation of evaluation:
\bean
(f(\theta_1) \ f(\theta_2) \ \cdots f(\theta_J)) G & = & f\left( (\theta_1 \ \theta_2 \cdots \theta_J)G\right) .
\eean
We have extended here the definition of $f(\cdot)$ to vectors through termwise application. Consider next, the matrix product
\bean
\Gamma & := & [\theta_1 \ \theta_2 \cdots \theta_J]G.
\eean
In writing this, we have abused notation and identified elements in $\mathbb{F}_{q^m}$ with their representations as vectors over $\mathbb{F}_q$ lying in $\mathbb{F}_q^m$.
The $m \times J$ matrix $[\theta_1 \ \theta_2 \cdots \theta_J]$ on the left has the property that all of its columns are linearly independent.  Hence linear dependence relations amongst columns of $\Gamma$ are precisely those inherited from the matrix $G$.  It follows that $\rho$ is also the smallest number, such that any $\rho$ thick columns of the product matrix $\Gamma$ have rank $\geq \tilde{J}$.  Since $f(\cdot)$ is uniquely determined by its evaluation at a collection of $\tilde{J}$ linearly independent vectors lying in $\mathbb{F}_q^m$, it follows that the maximum number of erasures that the code $\mathcal{C}$ can recover from is given by $n-\rho$.  Then, we have
\bean
D_{\min} & = & n-\rho+1.
\eean
\end{proof}

\vspace*{0.05in}

\begin{proof}  (of Thm.~\ref{thm:all_symbol_MBR})
Let $G_{\text{\tiny BASIC}}$ be the generator matrix of the code that is simply the disjoint union of the $t$ MBR codes.  As
it was explained previously,  the smallest number $\rho$ of thick columns of $G_{\text{\tiny BASIC}}$ such that any $\rho$
columns of $G_{\text{\tiny BASIC}}$ have rank $\geq K$ is given by $P^{(\text{inv})}(K)$.  It follows therefore from
Lemma~\ref{lem:useful} (by substituting $\tilde{J} = K$, $J = tK_L$ and also assuming that $G_{\text{\tiny BASIC}}$ is over
$\mathbb{F}_q$) that the code has minimum distance given by
\bean
d_{\min} & = & n \ - \ P^{\text{(inv)}} (K) +1,
\eean
hence the code attains the bound of Theorem~\ref{thm:URA_bound}, and thus, optimal.
\end{proof}

\begin{note}
 We note that whenever $K= v_1K_L + v_0, \ v_1 \geq 0, \ 1 \leq v_0 \leq K_L$ is such that $v_0 = \nu \alpha - {\nu \choose 2}\beta$ for some $1 \leq \nu \leq r$, then the code constructed by Construction~\ref{constr:all_symbol_MBR} has maximum possible scalar dimension given in~\eqref{eq:FileSizeBound-MBR}. This observation holds for the code we will construct using Construction~\ref{constr:info_MBR} as well.
\end{note}


\begin{constr}\label{constr:info_MBR}
We describe here a method by which we construct a code of length $n=tn_L+\Delta$, with $(r,\delta)$ \emph{information} locality for scalar dimension $K \leq t K_L$.
Given the message vector ${\bf u} \in \mathbb{F}_{q^m}^K$, we first encode ${\bf u}$ to a $tK_L+\Delta\alpha$ long Gabidulin codeword using
a $[tK_L+\Delta\alpha,K,tK_L+\Delta\alpha-K+1]_{q^m}$ Gabidulin code, for $m \geq tK_L+\Delta\alpha$.
We then divide the first $tK_L$ symbols of the Gabidulin codeword into $t$ disjoint groups of equal size and encode each of these $t$ groups using an $((n_L,r,d),(\alpha,\beta),K_L)$ MBR code (similar to the second step of encoding in Construction~\ref{constr:all_symbol_MBR}). This gives us a code of length $tn_L$ with MBR all-symbol locality, whose elements are $\left\{ {\bf c}^{\rm mbr}_i \mid i=1,2,\ldots,t\right\}$.  We then partition the remaining $\Delta\alpha$ symbols of the Gabidulin codeword into $\Delta$ equal sets and denote the $i^\text{th}$ set by ${\bf c}_{tn_L + i}$. The construction outputs $({\bf c}^{\rm mbr}_1,\ldots, {\bf c}^{\rm mbr}_t, {\bf c}_{tn_L + 1},\ldots, {\bf c}_{tn_L + \Delta})$ as a final codeword. The resultant vector code  $\mathcal{C}$ has:
\begin{itemize}
 \item Length $n=tn_L+\Delta$
 \item $t$ local $((n_L,r,d),(\alpha,\beta), K_L)$  MBR  codes with disjoint support
 \item full $(r,\delta)$ information locality
\end{itemize}
\end{constr}

\begin{thm}\label{thm:MBR_general}

Given any set of parameters $n,r,\delta,K$, such that $n=tn_L+\Delta$ and $K \leq tK_L$,  Construction \ref{constr:all_symbol_MBR} results in an optimal MBR-local code with $(r,\delta)$ information locality whose minimum distance is given by
\bean
d_{\min} =  n \ - \ P^{\text{(inv)}} (K) +1.
\eean
\end{thm}

\begin{proof}
The proof follows along the same lines as the proof of Theorem \ref{thm:all_symbol_MBR}.
\end{proof}

\section{Fractional-Repetition Codes as Local Codes}

In this section, we discuss the usage of fractional repetition (FR) codes as local codes in Constructions~\ref{constr:all_symbol_MBR} and~\ref{constr:info_MBR}.   FR codes can be viewed as a generalization of repair-by-transfer MBR codes, where a  repair process is \emph{uncoded} and \emph{table-based}, i.e., FR codes have a "repair-by-transfer" property, while only specific sets of nodes of size $d$ participate in a node repair process.  For the sake of completeness, we provide an overview of the $t$-design-based construction for FR codes presented in~\cite{ElrRam}\footnote{The construction in \cite{ElrRam} sets $t=2$ and $\lambda=1$; and the corresponding codes are called transposed codes.}.


Let $t,n,w,\lambda$ be integers with $n>w\geq t$ and $\lambda>0$.
A $t$-$(n,w,\lambda)$ design is a collection $\Bc$ of $w$-subsets (the \emph{blocks}), of an $n$-set $\Xc$ (the \emph{points}), such that every $t$-subset of $\Xc$ is contained in exactly $\lambda$ blocks. Let $x_1,\ldots, x_t\in\Xc$ be a set of $t$ points. We denote by $\lambda_s$ the number of blocks containing $x_1,\ldots,x_s$, $1\leq s\leq t$. Then,
\[
\lambda_{s}  =  \lambda \frac{{n-s \choose t-s}}{{w-s \choose t-s}};
\]
the  number of blocks in the $t$-design is  $b=\lambda_0=\lambda {n \choose t}/{w \choose t} $; and each point in $\Xc$ is contained in $\lambda_1$ blocks where $\lambda_1=\lambda {n-1 \choose t-1}/{w-1 \choose t-1}$~\cite{MacSlo}.

\begin{constr}\label{constr:FR}
Let $B_1,\ldots,B_b\in \Bc$ be the blocks and $x_1,\ldots, x_n\in \Xc$ the points of a  $t$-$(n,w,\lambda)$ design. Then the $n$ nodes of a FR code $C$ are given by the points of the design, i.e., a node $N_i$ contains  $\alpha\triangleq\lambda_1$ symbols given by  $N_i=\{j: x_i\in B_j\}$. Note that the cardinality of an intersection of any $s\leq t$ nodes are given by the numbers $\lambda_s$, and hence the cardinality of a union of any $s\leq t$ nodes can be easily derived by the inclusion-exclusion formula. Let $k,K$ be two integers such that $k\leq t$ and

\begin{equation}
\label{eq:frac_rep_URA}
|\bigcup_{i=1}^{k-1}N_i|<K\leq |\bigcup_{i=1}^{k}N_i|.
\end{equation}

Then we have an FR code over an alphabet of size $b$, with the property that there \emph{exists} a set of $d$ nodes which
can repair a failed node and from \emph{any} set of $k$ nodes one can reconstruct the original $K$ symbols.

Given a message vector $[m_1 \ m_2 \ \cdots m_K]$, we encode the message symbols first by using an $[b,K,b-K+1]$ MDS code to
produce $b$ coded symbols $(c_1, \ c_2 \cdots c_b)$ and then by employing the FR code based on the $t$-design to produce $n$
nodes  each containing $\lambda_1$ symbols.
\end{constr}

This family of FR codes based on $t$-designs is also an example of codes with uniform rank accumulation, and thus the bound of Theorem~\ref{thm:URA_bound} can be used here as well.  Thus, we have the following result.

\begin{thm}
 When FR codes based on a $t$-design obtained by Construction~\ref{constr:FR} are used as the local codes in Constructions~\ref{constr:all_symbol_MBR} and~\ref{constr:info_MBR}, then the resulting code with locality attains the bound of Theorem~\ref{thm:URA_bound} on minimum distance.
\end{thm}

\begin{figure} [t!]
\begin{center}
\includegraphics[width=3.8in]{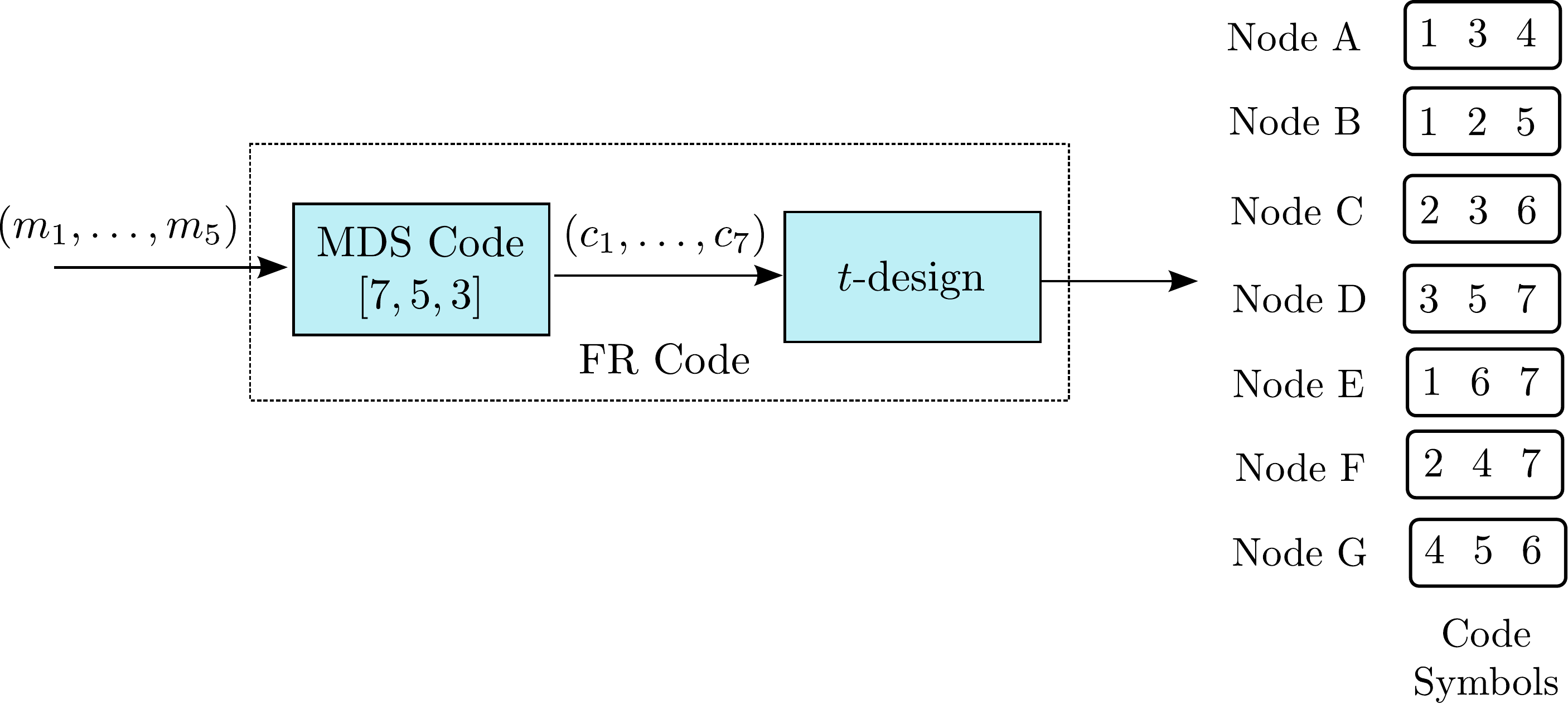}
\end{center}
\caption{Fractional Repetition Code based on $2-(7,3,1)$ design.}
\label{fig:fractional_repitition_code}
\end{figure}

\begin{figure} [t!]
\begin{center}
\includegraphics[height=1.5in]{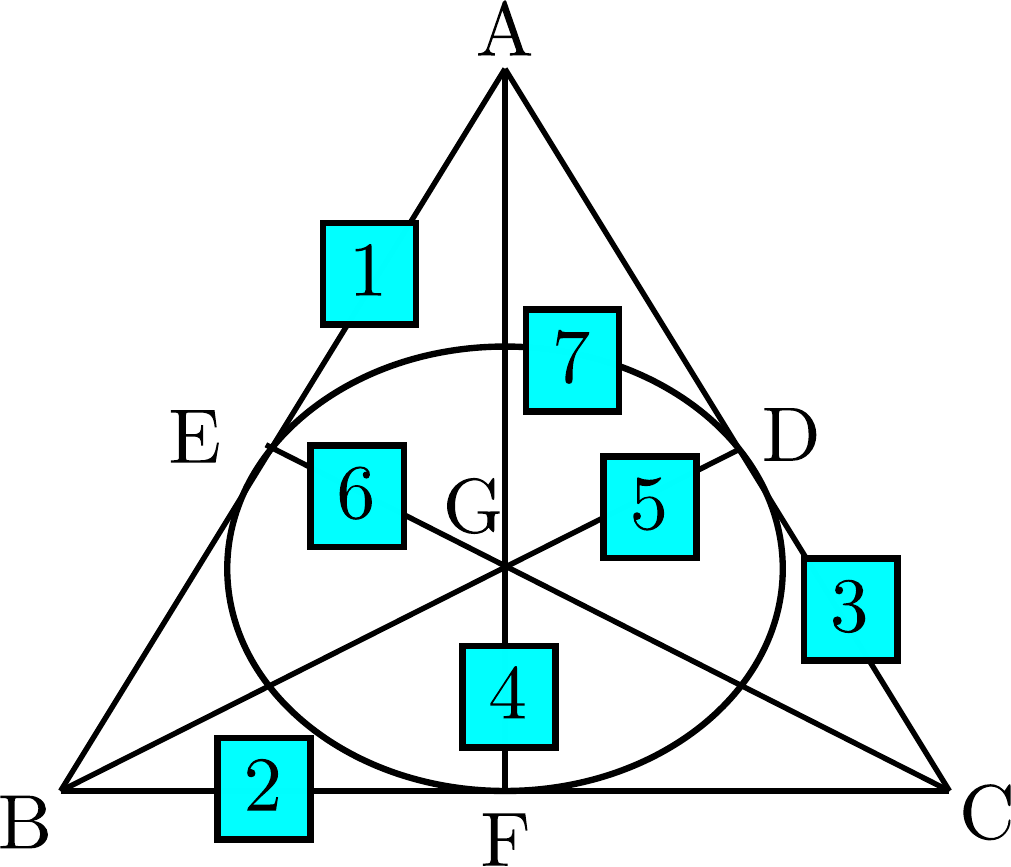}
\end{center}
\caption{Fano Plane, a $2-(7,3,1)$ design.}
\label{fig:transpose_plane}
\end{figure}

An example  of an encoding is shown in Fig.~\ref{fig:fractional_repitition_code}, where the encoding is done using $2$-$(7,3,1)$ design, also known as the Fano plane (see Fig.~\ref{fig:transpose_plane}). When we replace a local  MBR code with the FR code based on Fano plane in Fig.~\ref{fig:mbr_allsymbol}, we obtain a code with locality which has the optimal minimum distance.



\bibliographystyle{IEEEtran}
\bibliography{localMBRjoint}

\begin{thebibliography}{10}
\providecommand{\url}[1]{#1}
\csname url@samestyle\endcsname
\providecommand{\newblock}{\relax}
\providecommand{\bibinfo}[2]{#2}
\providecommand{\BIBentrySTDinterwordspacing}{\spaceskip=0pt\relax}
\providecommand{\BIBentryALTinterwordstretchfactor}{4}
\providecommand{\BIBentryALTinterwordspacing}{\spaceskip=\fontdimen2\font plus
\BIBentryALTinterwordstretchfactor\fontdimen3\font minus
  \fontdimen4\font\relax}
\providecommand{\BIBforeignlanguage}[2]{{%
\expandafter\ifx\csname l@#1\endcsname\relax
\typeout{** WARNING: IEEEtran.bst: No hyphenation pattern has been}%
\typeout{** loaded for the language `#1'. Using the pattern for}%
\typeout{** the default language instead.}%
\else
\language=\csname l@#1\endcsname
\fi
#2}}
\providecommand{\BIBdecl}{\relax}
\BIBdecl

\bibitem{GopHuaSimYek}
P.~Gopalan, C.~Huang, H.~Simitci, and S.~Yekhanin, ``{O}n the locality of
  codeword symbols,'' \emph{IEEE Trans. Inf. Theory}, vol.~58, no.~11, pp.
  6925--6934, Nov. 2012.

\bibitem{PraKamLalKum}
N.~Prakash, G.~M. Kamath, V.~Lalitha, and P.~V. Kumar, ``{O}ptimal linear codes
  with a local-error-correction property,'' in \emph{Proc. IEEE Int. Symp. Inf.
  Theory (ISIT)}, Cambridge, MA, Jul. 2012, pp. 2776--2780.

\bibitem{PapDim}
D.~S. Papailiopoulos and A.~G. Dimakis, ``{L}ocally repairable codes,'' in
  \emph{Proc. IEEE Int. Symp. Inf. Theory (ISIT)}, Cambridge, MA, Jul. 2012,
  pp. 2771--2775.

\bibitem{KamPraLalKum}
\BIBentryALTinterwordspacing
G.~M. {Kamath}, N.~{Prakash}, V.~{Lalitha}, and P.~V. {Kumar}, ``{C}odes with
  local regeneration,'' Nov. 2012. [Online]. Available: \url{arXiv:1211.1932}
\BIBentrySTDinterwordspacing

\bibitem{RawKoySilVis}
\BIBentryALTinterwordspacing
A.~S. Rawat, O.~O. Koyluoglu, N.~Silberstein, and S.~Vishwanath, ``{O}ptimal
  locally repairable and secure codes for distributed storage systems,'' Oct.
  2012. [Online]. Available: \url{arXiv:1210.6954}
\BIBentrySTDinterwordspacing

\bibitem{DimGodWuWaiRam}
A.~G. Dimakis, P.~B. Godfrey, Y.~Wu, M.~J. Wainwright, and K.~Ramchandran,
  ``{N}etwork coding for distributed storage systems,'' \emph{IEEE Trans. Inf.
  Theory}, vol.~56, no.~9, pp. 4539--4551, Sep. 2010.

\bibitem{RasShaKum_pm}
K.~V. Rashmi, N.~B. Shah, and P.~V. Kumar, ``{O}ptimal exact-regenerating codes
  for distributed storage at the {MSR} and {MBR} points via a product-matrix
  construction,'' \emph{IEEE Trans. Inf. Theory}, vol.~57, no.~8, pp.
  5227--5239, Aug. 2011.

\bibitem{ShaRasKumRam_rbt}
N.~B. Shah, K.~V. Rashmi, P.~V. Kumar, and K.~Ramchandran, ``{D}istributed
  storage codes with repair-by-transfer and nonachievability of interior points
  on the storage-bandwidth tradeoff,'' \emph{IEEE Trans. Inf. Theory}, vol.~58,
  no.~3, pp. 1837--1852, Mar. 2012.

\bibitem{MacSlo}
F.~J. MacWilliams and N.~J.~A. Sloane, \emph{{T}he Theory of Error-Correcting
  Codes}.\hskip 1em plus 0.5em minus 0.4em\relax Amsterdam: North-Holland,
  1983.

\bibitem{Gab}
E.~M. Gabidulin, ``{T}heory of codes with maximum rank distance,''
  \emph{Problems of Information Transmission}, vol.~21, pp. 1--12, July 1985.

\bibitem{OggDat}
F.~Oggier and A.~Datta, ``{S}elf-repairing homomorphic codes for distributed
  storage systems,'' in \emph{Proc. IEEE Int. Conf. Comput. Communications
  (INFOCOM)}, Shanghai, China, Apr. 2011, pp. 1215--1223.

\bibitem{SilRawVis}
N.~Silberstein, A.~S. Rawat, and S.~Vishwanath, ``{E}rror resilience in
  distributed storage via rank-metric codes,'' in \emph{Proc. 50th Annual
  Allerton Conf. on Communication, Control, and Computing (Allerton)},
  Urbana-Champaign, IL, Oct. 2012, pp. 1150 --1157.

\bibitem{ElrRam}
S.~El~Rouayheb and K.~Ramchandran, ``{F}ractional repetition codes for repair
  in distributed storage systems,'' in \emph{Proc. 48th Annual Allerton Conf.
  on Communication, Control, and Computing (Allerton)}, Urbana-Champaign, IL,
  Sep. 2010, pp. 1510 --1517.

\end{thebibliography}

\end{document}